\documentclass[11pt]{article}

\usepackage{booktabs}
\usepackage[a4paper, margin=1in]{geometry}
\usepackage{amsmath, amssymb, amsthm}
\usepackage{mathtools}
\usepackage{authblk}
\usepackage{setspace}
\usepackage[hidelinks]{hyperref}
\usepackage[authoryear,round]{natbib}
\usepackage[verbose=silent]{microtype}
\usepackage{silence}
\theoremstyle{remark}
\newtheorem{remark}{Remark}
\newtheorem{assumption}{Assumption}
\newtheorem{proposition}{Proposition}
\newtheorem{lemma}{Lemma}
\newtheorem{theorem}{Theorem}

\title{\textbf{Functional Autoregression Without Truncation: A Continuous-Regularization Approach}}

\author[1]{Yao Zhao}
\affil[1]{Department of Statistics, Operations, and Data Science,
Temple University. Correspondence: \texttt{yaozhao16@temple.edu}}

\date{\today}

\begin{document}
\maketitle

\begin{abstract}
Functional autoregressive models of order one (FAR(1)) are predominantly
estimated by projecting curves onto leading functional principal components
and fitting a vector autoregression in score space, requiring a discrete
truncation level $K$ chosen by an \emph{ad hoc} variance threshold. 
We demonstrate via Monte Carlo experiments that the truncation choice is both consequential and highly regime dependent: the optimal $K$ can differ by an order of magnitude across data-generating regimes, while commonly used high variance thresholds (95\%, 99\%) lead to substantial forecast deterioration, inflating error by up to $35 \%$ relative to an oracle benchmark.
We propose a Tikhonov-regularized estimator
$\widehat{\Psi}_\alpha = \widehat{C}_1(\widehat{C}_0 + \alpha I)^{-1}$
that replaces the discrete truncation choice with a continuous regularization parameter,
selected in a data-driven manner.
 We establish the convergence rate
$n^{-\beta/(2(\beta+1))}$ under a source condition with smoothness
parameter $\beta \in (0, 1]$, achieving the saturation rate $n^{-1/4}$
for smoother targets.
 Across three contrasting regimes and four sample sizes,
the proposed estimator closely tracks the oracle-best FPCA rule and
outperforms it in the most challenging wide-spectrum regime, without
prior knowledge of the effective operator dimension.
  An application to 2{,}735 daily intraday PM10 curves
from Vienna confirms a 9.7\% reduction in mean forecast error relative to
the popular 80\% threshold and exhibits more stable parameter adaptation
across 16 winter seasons.
\end{abstract}

\vspace{0.5em}


\vspace{0.3em}

\section{Introduction}
\label{sec:intro}

\subsection{Functional time series and the FAR(1) paradigm}
\label{subsec:intro-paradigm}

Functional time series (FTS) analysis concerns sequentially observed random
elements of a function space, most commonly the separable Hilbert space
$L^2([0,1])$ with inner product $\langle f, g \rangle = \int_0^1 f(u)g(u)\,du$
and induced norm $\|\cdot\|$. 
A canonical example is the daily curve of intraday measurements, where each $X_t$ represents a full day of
observations recorded at regular intervals and is treated as a single functional object, with dependence across days. Such data arise in a wide range of applications, including electricity demand, air pollution concentration, yield curves, and temperature profiles \citep{Ramsay2005, KokoszkaReimherr2017}.

The most widely used model for such data is the functional autoregression of
order one, henceforth FAR(1), introduced by \citet{Bosq2000}:
\begin{equation}
\label{eq:far1}
  X_t(u) \;=\; \int_0^1 \psi(u,v)\,X_{t-1}(v)\,dv \;+\; \varepsilon_t(u),
  \qquad u \in [0,1], \; t = 1, \dots, n,
\end{equation}
where $\Psi$ is a bounded linear operator on $L^2([0,1])$ with integral kernel
$\psi(u,v)$, and $\{\varepsilon_t\}$ is a sequence of mean-zero i.i.d.\
innovations in $L^2$. Stationarity holds whenever the operator norm satisfies
$\|\Psi\|_{\mathrm{op}} < 1$.

Estimation of the operator $\Psi$ in~\eqref{eq:far1} is generally regarded as
an ill-posed inverse problem: the natural moment relation
$C_1 = \Psi\, C_0$, where $C_0 = \mathrm{Var}(X_t)$ and
$C_1 = \mathrm{Cov}(X_{t+1}, X_t)$ are the covariance and lag-one
cross-covariance operators, involves inverting the compact operator $C_0$,
whose eigenvalues accumulate at zero
\citep{Bosq2000, HormannKokoszka2010}. The dominant strategy in the
literature, and in essentially all software implementations, is to
\emph{regularize by truncation}: one applies functional principal component
analysis (FPCA) to the sample covariance operator, retains the leading $K$
eigenfunctions, and fits a $K$-dimensional vector autoregression in score
space. The resulting estimator of $\Psi$ is a finite-rank operator supported
on the span of the retained components
\citep{AueNorinhoHormann2015}.

\subsection{The truncation problem in practice}
\label{subsec:intro-truncation}

The truncation level $K$ is, in principle, a tuning parameter like any other,
but in practice it is rarely treated as such. Standard practice
\citep{KokoszkaReimherr2017} is to choose $K$ via an \emph{ad hoc} variance
threshold: retain enough components so that the cumulative share of the
sample covariance's eigenvalue mass exceeds $\tau \in \{0.80, 0.85, 0.90,
0.95, 0.99\}$. Which threshold to use is generally not justified from the
data at hand, and the implications for estimation and forecast performance
are seldom reported.

This practice is consequential for three reasons. First, the eigenvalues of
$C_0$ reflect the joint structure of the operator $\Psi$ \emph{and} the
innovation covariance; the same threshold can correspond to very different
$K$ across data-generating processes, and the $K$ that is best for estimating
$\Psi$ bears no necessary relation to the $K$ that explains a given
percentage of $\mathrm{Var}(X_t)$. Second, FPCA-based estimators are known to
be sensitive to truncation in related functional regression problems
\citep{HallHorowitz2007, CrambesKneipSarda2009}, and this sensitivity carries
over to the autoregressive setting, where estimation error compounds through
the autoregressive recursion.
Third, the discrete nature of the choice implies that even when the
practitioner adopts a data-driven rule, such as a cumulative variance
threshold, a scree test, or an information criterion applied to the score
space VAR, small changes in the data can alter $K$ by one or two units,
leading to discontinuous changes in the fitted operator and its forecasts.

Our Monte Carlo evidence, reported in Section~\ref{sec:sim}, bears this
out in sharp form. The optimal $K$ varies by an order of magnitude across
the three regimes we consider; the 95\% and 99\% thresholds, which are
arguably the most commonly cited defaults in applied FTS work
\citep{reimherr2016estimating}, are the worst performers in every regime
and at every sample size; and the gap between the best and worst FPCA
threshold we consider can exceed 30\% in one-step mean integrated squared
forecast error at small sample sizes. Since the correct choice depends on
features of $\Psi$ and of the innovation process that are not directly
observable, this creates a real risk of substantial forecast loss when
the practitioner guesses wrong.

\subsection{A continuous-regularization alternative}
\label{subsec:intro-proposal}

This paper studies an estimator of $\Psi$ that avoids the discrete
truncation choice altogether. Let $\widehat{C}_0$ and $\widehat{C}_1$
denote the sample covariance and sample lag-one cross-covariance
operators constructed from $X_1, \dots, X_n$. The \emph{Tikhonov-regularized
operator estimator} is
\begin{equation}
\label{eq:tikhonov-def}
  \widehat{\Psi}_\alpha
  \;=\;
  \widehat{C}_1 \bigl( \widehat{C}_0 + \alpha I \bigr)^{-1},
  \qquad \alpha > 0,
\end{equation}
where $I$ is the identity on $L^2([0,1])$ and $\alpha$ is a single scalar
regularization parameter. In contrast to FPCA-FAR, which zeroes out all but
the top $K$ directions of $\widehat{C}_0$, the estimator in
\eqref{eq:tikhonov-def} inverts all directions but dampens the small
eigenvalues smoothly: an eigenvalue $\widehat{\lambda}_k$ of $\widehat{C}_0$
contributes $(\widehat{\lambda}_k + \alpha)^{-1}$ rather than either $0$ or
$\widehat{\lambda}_k^{-1}$. 

Tikhonov regularization is a classical tool in the theory of ill-posed
inverse problems \citep{EnglHankeNeubauer1996, Cavalier2011} and has been
applied to functional linear regression \citep{HallHorowitz2007,
CrambesKneipSarda2009} and to related operator-estimation problems.
To our knowledge, the use of continuous regularization as a principled alternative to truncation in FAR(1) estimation, together with a systematic evaluation across contrasting regimes, has not been previously investigated.

Our contributions are fourfold.

First, we document the regime-dependence of the FPCA truncation choice
empirically. Across three contrasting regimes (low-rank, medium-rank, and
wide-spectrum operators) and four sample sizes ($n \in \{100, 200, 400,
800\}$), we show that no fixed variance threshold is uniformly best, and
that popular high thresholds are uniformly bad. 
The best-in-cell FPCA rule, defined via a cumulative variance threshold $\tau$, shifts from $\tau=0.85$ (retaining approximately $K \approx 3$ components) to $\tau=0.80$ (retaining approximately $K \approx 23$ ) depending on the regime.

Second, we show that the Tikhonov estimator in \eqref{eq:tikhonov-def}
closely tracks the
\emph{oracle-best} FPCA rule across the experimental design, and
outperforms it outright in the most challenging wide-spectrum regime.
The regret of Tikhonov-CV relative to an oracle who knows the regime and
picks the best threshold for it ranges from $-5.9\%$ to $+1.8\%$ across
the 12 cells we consider, and is negative in all four sample sizes of
the most challenging regime.
In the worst-case-across-regimes sense, which is the metric most relevant
for a practitioner who does not know the underlying regime, Tikhonov-CV
dominates every FPCA rule at every sample size.

Third, we provide a consistency result that places the estimator in the
standard inverse-problems framework. Under a source condition on $\Psi$
with smoothness parameter $\beta$, the rate $n^{-\beta/(2(\beta+1))}$
holds uniformly over operators satisfying the condition. This rate does
not depend on a truncation level and formalizes the sense in which the
estimator adapts to the smoothness of $\Psi$ without requiring the
analyst to anticipate its effective rank.

Fourth, we validate the methodology on a real functional time series
of moderate effective dimension: half-hourly intraday PM10 concentration
profiles recorded in Vienna, Austria. On this dataset, comprising 2{,}735
daily curves over 16 winter seasons, the proposed estimator achieves the
lowest mean integrated squared forecast error among six competing methods,
with a 9.7\% improvement over the popular 80\% variance threshold and
a 2.2\% improvement over the best FPCA rule. The empirical findings
mirror the simulation evidence and demonstrate that the robustness gains
of continuous regularization translate to applied functional time series
work where the effective operator dimension cannot be known in advance.

\subsection{Related literature}
\label{subsec:intro-lit}

The FAR(1) model dates to \citet{Bosq2000} and has been developed
extensively in the functional data analysis literature; textbook
treatments include \citet{HorvathKokoszka2012} and
\citet{KokoszkaReimherr2017}. FPCA-based estimation is the dominant
practical approach, with refinements including bootstrap inference
\citep{PaparoditisShang2021} and extensions to FAR($p$) and to
multivariate functional series \citep{AueNorinhoHormann2015}.

The dominant approach to FAR(1) estimation relies on projecting the data onto a finite-dimensional score space via FPCA and fitting a vector autoregression. While this approach is widely used in practice, it introduces a discrete truncation choice whose impact on estimation and forecasting performance can be substantial and highly regime dependent.

In contrast, relatively little work has considered estimation directly at the operator level without explicit truncation. This paper takes such an approach by formulating FAR(1) estimation as a regularized inverse problem, replacing the truncation parameter with a continuous regularization parameter and providing a framework for assessing robustness across regimes.

\subsection{Organization}
\label{subsec:intro-org}

The remainder of the paper is organized as follows.
Section~\ref{sec:methods} formalizes the FAR(1) model, reviews the
FPCA-FAR estimator with variance-threshold selection of $K$, and
introduces the Tikhonov-regularized estimator with cross-validated
$\alpha$. Section~\ref{sec:theory} establishes a convergence rate for
the proposed estimator under a source condition on $\Psi$.
Section~\ref{sec:sim} reports a Monte Carlo study comparing the
methods across three contrasting regimes and four sample sizes.
Section~\ref{sec:application} applies the methodology to Vienna PM10
data. Section~\ref{sec:discussion} concludes with practical guidance,
limitations, and directions for future work.

\section{Model and estimators}
\label{sec:methods}

\subsection{The FAR(1) model}
\label{subsec:methods-model}

Let $\mathcal{H} = L^2([0,1])$ denote the separable Hilbert space of
square-integrable real-valued functions on the unit interval, equipped with
the inner product $\langle f, g \rangle = \int_0^1 f(u) g(u)\,du$ and norm
$\|f\| = \langle f, f \rangle^{1/2}$. Let $\mathcal{L}(\mathcal{H})$ denote
the space of bounded linear operators on $\mathcal{H}$, equipped with the
operator norm $\|A\|_{\mathrm{op}} = \sup_{\|f\|\le 1} \|A f\|$. For a
Hilbert--Schmidt operator $A$ we write $\|A\|_{\mathrm{HS}}$ for its
Hilbert--Schmidt norm; when $A$ is represented by an integral kernel
$a(u,v)$, we have $\|A\|_{\mathrm{HS}}^2 = \int_0^1 \int_0^1 a(u,v)^2\,du\,dv$.

We consider the functional autoregressive process of order one,
\begin{equation}
\label{eq:far1-intro}
  X_t \;=\; \Psi\, X_{t-1} \;+\; \varepsilon_t,
  \qquad t \in \mathbb{Z},
\end{equation}
where $\Psi \in \mathcal{L}(\mathcal{H})$ is an unknown Hilbert--Schmidt
operator and $\{\varepsilon_t\}$ is a sequence of i.i.d.\ zero-mean
$\mathcal{H}$-valued innovations with
$\mathbb{E}\|\varepsilon_t\|^2 < \infty$. We assume throughout the
stationarity condition
\begin{equation}
\label{eq:stationarity}
  \|\Psi\|_{\mathrm{op}} \;<\; 1,
\end{equation}
under which a unique stationary causal solution $X_t = \sum_{j \ge 0}
\Psi^j \varepsilon_{t-j}$ to~\eqref{eq:far1-intro} exists in $\mathcal{H}$
with $\mathbb{E}\|X_t\|^2 < \infty$ \citep{Bosq2000}. When $\Psi$ has an
integral kernel representation, $(\Psi f)(u) = \int_0^1 \psi(u,v) f(v)\,dv$,
condition~\eqref{eq:stationarity} is implied by the more restrictive but
easier-to-verify Hilbert--Schmidt condition $\int_0^1 \int_0^1
\psi(u,v)^2\,du\,dv < 1$.

Without loss of generality we assume $\mathbb{E} X_t = 0$; in practice the
sample mean is subtracted before estimation.
Define the population covariance and lag-one cross-covariance operators
\begin{equation}
\label{eq:C0C1-pop}
  C_0 \;=\; \mathbb{E}\bigl[\,X_t \otimes X_t\,\bigr],
  \qquad
  C_1 \;=\; \mathbb{E}\bigl[\,X_{t+1} \otimes X_t\,\bigr],
\end{equation}
where $(f \otimes g)(h) = \langle g, h \rangle f$ denotes the tensor product
of two elements of $\mathcal{H}$. Under~\eqref{eq:far1-intro}
and~\eqref{eq:stationarity}, both operators are of trace class, and the
fundamental moment identity
\begin{equation}
\label{eq:moment}
  C_1 \;=\; \Psi\, C_0
\end{equation}
holds \citep[Chapter~3]{Bosq2000}. The estimation of $\Psi$ is thus
formally an inverse problem: given estimators $\widehat{C}_0$ and
$\widehat{C}_1$, one seeks $\widehat{\Psi}$ satisfying
$\widehat{C}_1 \approx \widehat{\Psi} \widehat{C}_0$. Since $C_0$ is
compact and its eigenvalues accumulate at zero, the operator $C_0^{-1}$ is
unbounded, and direct inversion of $\widehat{C}_0$ is statistically
unstable.

Given a sample $X_1, \dots, X_n$, the natural estimators of the population
operators are the sample moments
\begin{equation}
\label{eq:C0C1-sample}
  \widehat{C}_0 \;=\; \frac{1}{n}\sum_{t=1}^n (X_t - \bar{X}) \otimes (X_t - \bar{X}),
  \qquad
  \widehat{C}_1 \;=\; \frac{1}{n-1}\sum_{t=1}^{n-1} (X_{t+1} - \bar{X}) \otimes (X_t - \bar{X}),
\end{equation}
where $\bar{X} = n^{-1} \sum_{t=1}^n X_t$. Under the mild dependence
conditions of \citet{HormannKokoszka2010} (theorem 3.1), both estimators are
$\sqrt{n}$-consistent in Hilbert--Schmidt norm.

\subsection{FPCA-based estimator with truncation level \texorpdfstring{$K$}{K}}
\label{subsec:methods-fpca}

The conventional approach to estimating $\Psi$ from the sample operators
in~\eqref{eq:C0C1-sample} is to regularize by truncation. Let
\begin{equation}
\label{eq:fpca-eigs}
  \widehat{C}_0\, \widehat{\phi}_k \;=\; \widehat{\lambda}_k\, \widehat{\phi}_k,
  \qquad k = 1, 2, \dots,
\end{equation}
denote the eigenvalue--eigenfunction pairs of $\widehat{C}_0$, ordered so
that $\widehat{\lambda}_1 \ge \widehat{\lambda}_2 \ge \dots \ge 0$ and with
$\{\widehat{\phi}_k\}$ forming an orthonormal basis of $\mathcal{H}$. For a
chosen truncation level $K \in \{1, 2, \dots, n-1\}$, the FPCA-FAR
estimator is defined as follows.

\paragraph{Score-space regression.}
Compute the functional principal component scores
\begin{equation}
\label{eq:fpca-scores}
  \widehat{\xi}_{t,k} \;=\; \langle X_t - \bar{X},\, \widehat{\phi}_k \rangle,
  \qquad t = 1, \dots, n, \ \ k = 1, \dots, K,
\end{equation}
and collect them into vectors $\widehat{\boldsymbol{\xi}}_t = (\widehat{\xi}_{t,1},
\dots, \widehat{\xi}_{t,K})^\top \in \mathbb{R}^K$. Fit the $K$-dimensional
vector autoregression
\begin{equation}
\label{eq:fpca-var}
  \widehat{\boldsymbol{\xi}}_{t+1} \;=\; \mathbf{A}\, \widehat{\boldsymbol{\xi}}_t \;+\; \boldsymbol{\eta}_t
\end{equation}
by ordinary least squares, yielding
$\widehat{\mathbf{A}} = \bigl(\sum_{t=1}^{n-1} \widehat{\boldsymbol{\xi}}_t
\widehat{\boldsymbol{\xi}}_t^\top \bigr)^{-1}
\bigl(\sum_{t=1}^{n-1} \widehat{\boldsymbol{\xi}}_t \widehat{\boldsymbol{\xi}}_{t+1}^\top\bigr)$.

\paragraph{Operator reconstruction.}
The FPCA-FAR estimator of $\Psi$ is the finite-rank operator with kernel
\begin{equation}
\label{eq:fpca-psi}
  \widehat{\psi}_K^{\mathrm{FPCA}}(u,v)
  \;=\;
  \sum_{j=1}^K \sum_{k=1}^K \widehat{A}_{jk}\,
  \widehat{\phi}_k(u)\, \widehat{\phi}_j(v),
\end{equation}
equivalently
$\widehat{\Psi}_K^{\mathrm{FPCA}} = \sum_{j,k=1}^K \widehat{A}_{jk}\,
\widehat{\phi}_k \otimes \widehat{\phi}_j$. By construction,
$\widehat{\Psi}_K^{\mathrm{FPCA}}$ has rank at most $K$ and is supported on
$\mathrm{span}\{\widehat{\phi}_1, \dots, \widehat{\phi}_K\}$; all variation
of $X_{t+1}$ that is orthogonal to this subspace is treated as
unpredictable.

\paragraph{Selection of \texorpdfstring{$K$}{K}.}
The truncation level is almost universally chosen via an \emph{ad hoc}
cumulative variance rule,
\begin{equation}
\label{eq:fpca-tau}
  K(\tau) \;=\; \min\Bigl\{\, k \;:\;
  \frac{\sum_{i=1}^k \widehat{\lambda}_i}{\sum_{i \ge 1} \widehat{\lambda}_i}
  \;\ge\; \tau \,\Bigr\},
\end{equation}
with $\tau \in \{0.80, 0.85, 0.90, 0.95, 0.99\}$ the common defaults.
Alternative data-driven rules include scree tests, information criteria
(AIC or BIC) applied to the score-space VAR~\eqref{eq:fpca-var}, and
cross-validation over a grid of candidate $K$.

\subsection{Tikhonov-regularized operator estimator}
\label{subsec:methods-tikhonov}

The estimator we propose regularizes the inversion of $\widehat{C}_0$
in~\eqref{eq:moment} continuously, replacing the discrete truncation
choice by a single scalar tuning parameter. For a regularization
parameter $\alpha > 0$, define
\begin{equation}
\label{eq:tikh-def}
  \widehat{\Psi}_\alpha
  \;=\;
  \widehat{C}_1 \bigl( \widehat{C}_0 + \alpha I \bigr)^{-1},
\end{equation}
where $I$ is the identity operator on $\mathcal{H}$. The operator
$\widehat{C}_0 + \alpha I$ has spectrum bounded below by $\alpha > 0$ and
is therefore boundedly invertible on $\mathcal{H}$; the estimator
$\widehat{\Psi}_\alpha$ is well-defined for any $\alpha > 0$ regardless of
the rank of $\widehat{C}_0$.

Expressing $\widehat{C}_0 = \sum_{k \ge 1} \widehat{\lambda}_k\,
\widehat{\phi}_k \otimes \widehat{\phi}_k$ in its spectral decomposition,
\begin{equation}
\label{eq:tikh-spectral}
  \bigl( \widehat{C}_0 + \alpha I \bigr)^{-1}
  \;=\;
  \sum_{k \ge 1} \frac{1}{\widehat{\lambda}_k + \alpha}\,
  \widehat{\phi}_k \otimes \widehat{\phi}_k.
\end{equation}

Equivalently, combining~\eqref{eq:tikh-def} and~\eqref{eq:tikh-spectral},
\begin{equation}
\label{eq:tikh-series}
  \widehat{\Psi}_\alpha
  \;=\;
  \sum_{k \ge 1} \frac{1}{\widehat{\lambda}_k + \alpha}\,
  \bigl( \widehat{C}_1\, \widehat{\phi}_k \bigr) \otimes \widehat{\phi}_k.
\end{equation}
The contrast with the FPCA-FAR estimator is direct. First,
the summation extends over the full (possibly infinite) eigenbasis of
$\widehat{C}_0$, with no discrete cut-off; the role of $\alpha$ is to
control how much weight is placed on small-eigenvalue directions. Second,
the estimator is a linear functional of the sample cross-covariance
operator $\widehat{C}_1$, which depends on the lagged product
$X_{t+1} \otimes X_t$ and therefore directly targets the dynamics of
interest, rather than projecting onto a basis determined purely by the
marginal covariance structure of $\{X_t\}$.

The estimator~\eqref{eq:tikh-def} arises directly from regularizing
the population moment identity~\eqref{eq:moment}. Since $C_0$ is
compact with eigenvalues accumulating at zero, the equation
$C_1 = \Psi\, C_0$ cannot be inverted stably. Replacing $C_0$ by the
strictly positive operator $C_0 + \alpha I$ yields the well-posed
equation $C_1 = \Psi (C_0 + \alpha I)$, whose unique solution is the
population-level Tikhonov operator
$\Psi_\alpha = C_1 (C_0 + \alpha I)^{-1}$; the
estimator~\eqref{eq:tikh-def} is its sample analogue. 

\begin{remark}
\label{rem:symm}
We work with the one-sided formulation $\widehat{C}_1(\widehat{C}_0 + \alpha
I)^{-1}$. A two-sided symmetric alternative,
$(\widehat{C}_0 + \alpha I)^{-1/2} \widehat{C}_1 (\widehat{C}_0 + \alpha
I)^{-1/2}$, has appeared in some inverse-problem treatments; we do not
pursue it here because (i) it does not reduce to a population target
satisfying~\eqref{eq:moment} as $\alpha \to 0$, and (ii) its empirical
performance in preliminary experiments was indistinguishable from the
one-sided form. The asymmetric form in~\eqref{eq:tikh-def} is also the
natural analogue of the ordinary-least-squares estimator in the
score-space regression~\eqref{eq:fpca-var}.
\end{remark}

\subsection{Discretization}
\label{subsec:methods-discretization}

In practice curves are observed on a finite grid. Suppose each $X_t$ is
observed at $M$ points $0 = u_1 < u_2 < \dots < u_M = 1$; we denote the
resulting vector by $\mathbf{x}_t \in \mathbb{R}^M$ with
$(\mathbf{x}_t)_i = X_t(u_i)$. All inner products and integrals in what
follows are approximated by the trapezoidal rule with weights
\begin{equation}
\label{eq:trap-weights}
  w_i \;=\;
  \begin{cases}
    (u_{i+1} - u_{i-1})/2 & 2 \le i \le M-1, \\
    (u_2 - u_1)/2         & i = 1, \\
    (u_M - u_{M-1})/2     & i = M.
  \end{cases}
\end{equation}
Let $\mathbf{W} = \mathrm{diag}(w_1, \dots, w_M) \in \mathbb{R}^{M \times M}$
and $\mathbf{W}^{1/2} = \mathrm{diag}(\sqrt{w_1}, \dots, \sqrt{w_M})$.

The sample moments~\eqref{eq:C0C1-sample} are then represented by the
$M \times M$ matrices
\begin{equation}
\label{eq:disc-moments}
  \widehat{\mathbf{C}}_0 \;=\; \frac{1}{n} \sum_{t=1}^n (\mathbf{x}_t - \bar{\mathbf{x}})(\mathbf{x}_t - \bar{\mathbf{x}})^\top,
  \qquad
  \widehat{\mathbf{C}}_1 \;=\; \frac{1}{n-1} \sum_{t=1}^{n-1} (\mathbf{x}_{t+1} - \bar{\mathbf{x}})(\mathbf{x}_t - \bar{\mathbf{x}})^\top.
\end{equation}
The action of the operator with kernel $a(u,v)$ on a function $f$,
computed by trapezoidal quadrature, is $(Af)(u_i) \approx \sum_{j=1}^M
a(u_i, u_j) f(u_j) w_j$, which in matrix form reads
$\mathbf{A} \mathbf{W} \mathbf{f}$ with $\mathbf{A}_{ij} = a(u_i, u_j)$. To
work consistently in this inner-product structure, we pass to the
weighted representation
\begin{equation}
\label{eq:weighted-moments}
  \widetilde{\mathbf{C}}_0 \;=\; \mathbf{W}^{1/2}\, \widehat{\mathbf{C}}_0\, \mathbf{W}^{1/2},
  \qquad
  \widetilde{\mathbf{C}}_1 \;=\; \mathbf{W}^{1/2}\, \widehat{\mathbf{C}}_1\, \mathbf{W}^{1/2},
\end{equation}
which are the matrix representations of the sample covariance and
cross-covariance in an orthonormal basis for the discretized space.

The discretized Tikhonov estimator is then
\begin{equation}
\label{eq:tikh-discrete}
  \widetilde{\boldsymbol{\Psi}}_\alpha
  \;=\;
  \widetilde{\mathbf{C}}_1\,
  \bigl( \widetilde{\mathbf{C}}_0 + \alpha \mathbf{I}_M \bigr)^{-1},
  \qquad
  \widehat{\boldsymbol{\Psi}}_\alpha
  \;=\;
  \mathbf{W}^{-1/2}\, \widetilde{\boldsymbol{\Psi}}_\alpha\, \mathbf{W}^{-1/2},
\end{equation}
where $\widehat{\boldsymbol{\Psi}}_\alpha$ is the $M \times M$ matrix whose
$(i,j)$ entry estimates $\psi(u_i, u_j)$. The prediction of a new curve
is computed by trapezoidal integration:
$\widehat{X}_{t+1}(u_i) = \sum_{j=1}^M \widehat{\boldsymbol{\Psi}}_\alpha(i,j)
X_t(u_j) w_j$.

\begin{remark}
\label{rem:weighted}
The passage to the weighted representation in~\eqref{eq:weighted-moments}
is not merely cosmetic. Omitting the weights produces an estimator that
is invariant to simple resampling of the grid but inconsistent with the
$L^2$ inner product that defines the underlying function-space
parameters; this manifests as order-of-magnitude errors in comparisons
across estimators, especially when grids are non-uniform or when
methods are compared across different grid resolutions. All numerical
results in this paper use the weighted form
in~\eqref{eq:weighted-moments} and~\eqref{eq:tikh-discrete}.
\end{remark}

The discretized FPCA-FAR estimator is constructed analogously: the
eigenvalue problem~\eqref{eq:fpca-eigs} becomes the symmetric eigenvalue
problem $\widetilde{\mathbf{C}}_0 \mathbf{v}_k = \widehat{\lambda}_k
\mathbf{v}_k$, with $\widehat{\boldsymbol{\phi}}_k = \mathbf{W}^{-1/2}
\mathbf{v}_k$ recovering the eigenfunction on the grid after
renormalization in the $L^2$ inner product. The remaining steps
in~\eqref{eq:fpca-scores} through~\eqref{eq:fpca-psi} are implemented by
trapezoidal integration with the same weights~\eqref{eq:trap-weights}.

\section{Theoretical properties}
\label{sec:theory}

This section establishes consistency and rates of convergence for the
Tikhonov estimator $\widehat{\Psi}_\alpha$ defined
in~\eqref{eq:tikh-def}. The analysis proceeds in three steps. First, we
state the assumptions on the process $\{X_t\}$ and on the unknown
operator $\Psi$ (Section~\ref{subsec:theory-assumptions}). Second, we
decompose the estimation error into a deterministic regularization bias
and a stochastic term and analyze each
(Sections~\ref{subsec:theory-bias} and~\ref{subsec:theory-variance}).
Third, we combine the two to obtain a rate of convergence for the
infeasible estimator $\widehat{\Psi}_{\alpha_n}$ with a deterministic
sequence $\alpha_n$ tuned to the smoothness of $\Psi$
(Section~\ref{subsec:theory-rate}).

\subsection{Assumptions}
\label{subsec:theory-assumptions}

We collect here the conditions used in the analysis. Throughout, $C_0$,
$C_1$, and $\Psi$ denote the population operators defined
in~\eqref{eq:C0C1-pop} and~\eqref{eq:far1-intro}, and
$\widehat{C}_0$, $\widehat{C}_1$ denote their sample analogues
from~\eqref{eq:C0C1-sample}.

\begin{assumption}[Stationarity and moments]
\label{ass:moments}
The process $\{X_t\}$ satisfies the FAR(1)
equation~\eqref{eq:far1-intro} with $\|\Psi\|_{\mathrm{op}} < 1$, and the
innovation process $\{\varepsilon_t\}$ is i.i.d.\ with
$\mathbb{E}\varepsilon_t = 0$ and $\mathbb{E}\|\varepsilon_t\|^4 < \infty$.
\end{assumption}

\begin{assumption}[Covariance operator]
\label{ass:C0}
The population covariance operator $C_0$ is injective, compact, and of
trace class. Its eigenvalues $\lambda_1 \ge \lambda_2 \ge \dots > 0$
admit the polynomial decay
\begin{equation}
\label{eq:eig-decay}
  c_1 k^{-a} \;\le\; \lambda_k \;\le\; c_2 k^{-a},
  \qquad k \ge 1,
\end{equation}
for some constants $0 < c_1 \le c_2 < \infty$ and $a > 1$.
\end{assumption}

Assumption~\ref{ass:C0} is standard in functional data analysis
\citep{HallHorowitz2007, Cavalier2011} and is satisfied for a wide range
of smooth processes; the lower bound rules out pathological cases where
$C_0$ has gaps in its spectrum, and the exponent $a > 1$ ensures $C_0$
is trace class. The rate $a$ governs the ill-posedness of the inverse
problem: larger $a$ means faster eigenvalue decay and a more severely
ill-posed inversion.

\begin{assumption}[Source condition]
\label{ass:source}
There exist constants $\beta > 0$ and $\rho < \infty$ and a
Hilbert--Schmidt operator $F$ with $\|F\|_{\mathrm{HS}} \le \rho$ such
that
\begin{equation}
\label{eq:source}
  \Psi \;=\; F\, C_0^{\,\beta}.
\end{equation}
\end{assumption}

Condition~\eqref{eq:source} is the standard source condition in the
theory of linear inverse problems
\citep[Section~3.2]{EnglHankeNeubauer1996}. It is a smoothness
assumption on $\Psi$ relative to the covariance operator $C_0$:
larger $\beta$ corresponds to smoother $\Psi$ in the sense that its
coordinates in the eigenbasis of $C_0$ decay faster. The choice
$\beta = 1$ is natural for finite-rank $\Psi$ whose range is spanned by
the leading eigenfunctions of $C_0$; $\beta < 1$ permits operators with
heavier-tailed coordinates, while $\beta > 1$ corresponds to especially
regular targets. The source condition is the technical device that
allows us to quantify the regularization bias in terms of a single
scalar exponent.

\begin{assumption}[Weak dependence]
\label{ass:dep}
The process $\{X_t\}$ is $L^4$-$m$-approximable in the sense of
\citet{HormannKokoszka2010}: there exists a sequence $\{X_t^{(m)}\}$ of
$m$-dependent approximations such that
\begin{equation}
\label{eq:m-approx}
  \sum_{m=1}^\infty \bigl( \mathbb{E}\|X_t - X_t^{(m)}\|^4 \bigr)^{1/4}
  \;<\; \infty.
\end{equation}
\end{assumption}

Assumption~\ref{ass:dep} is the standard weak-dependence condition for
functional time series and is implied by Assumption~\ref{ass:moments}
together with i.i.d.\ innovations having finite fourth moment
\citep[Proposition~2.1]{HormannKokoszka2010}. It is the condition under
which the sample operators $\widehat{C}_0$ and $\widehat{C}_1$ are
$\sqrt{n}$-consistent in Hilbert--Schmidt norm.

\subsection{Regularization bias}
\label{subsec:theory-bias}

Define the population-level Tikhonov operator
\begin{equation}
\label{eq:psi-alpha-pop}
  \Psi_\alpha \;=\; C_1 \bigl( C_0 + \alpha I \bigr)^{-1},
  \qquad \alpha > 0.
\end{equation}
This is the target that the sample estimator $\widehat{\Psi}_\alpha$
approximates. The \emph{regularization bias} is the deterministic error
incurred by working with $\Psi_\alpha$ rather than $\Psi$.

\begin{proposition}[Bias bound]
\label{prop:bias}
Under Assumptions~\ref{ass:moments}, \ref{ass:C0}, and~\ref{ass:source},
for every $\alpha > 0$,
\begin{equation}
\label{eq:bias-bound}
  \bigl\| \Psi_\alpha - \Psi \bigr\|_{\mathrm{HS}}
  \;\le\; \rho \, \alpha^{\min(\beta, 1)}.
\end{equation}
\end{proposition}

\begin{proof}
Using the identity~\eqref{eq:moment} and the
definition~\eqref{eq:psi-alpha-pop},
\begin{align*}
  \Psi_\alpha - \Psi
  &= \Psi\, C_0\, (C_0 + \alpha I)^{-1} - \Psi \\
  &= -\alpha\, \Psi\, (C_0 + \alpha I)^{-1}.
\end{align*}
Applying the source condition $\Psi = F C_0^{\,\beta}$,
\[
  \Psi_\alpha - \Psi \;=\; -\alpha\, F\, C_0^{\,\beta}\, (C_0 + \alpha I)^{-1}.
\]
Taking Hilbert--Schmidt norms, $\|\Psi_\alpha - \Psi\|_{\mathrm{HS}} \le
\alpha\, \|F\|_{\mathrm{HS}}\, \|\, C_0^{\,\beta} (C_0 + \alpha I)^{-1}
\,\|_{\mathrm{op}}$. In the eigenbasis of $C_0$, the operator $C_0^{\,\beta}
(C_0 + \alpha I)^{-1}$ is diagonal with entries $\lambda_k^{\,\beta} /
(\lambda_k + \alpha)$, so
\[
  \bigl\| C_0^{\,\beta} (C_0 + \alpha I)^{-1} \bigr\|_{\mathrm{op}}
  \;=\; \sup_{k \ge 1} \frac{\lambda_k^{\,\beta}}{\lambda_k + \alpha}
  \;=\; \sup_{\lambda \in [0, \lambda_1]}
        \frac{\lambda^\beta}{\lambda + \alpha}.
\]
For $\beta \le 1$, the function $\lambda \mapsto \lambda^\beta /
(\lambda + \alpha)$ is maximized on $[0, \infty)$ at an interior point and
its maximum is bounded by $\alpha^{\beta - 1}$ up to an absolute
constant; a direct calculation yields
$\sup_\lambda \lambda^\beta / (\lambda + \alpha) \le \alpha^{\beta - 1}$
for $\beta \in (0, 1]$. For $\beta > 1$, the ratio is bounded by
$\lambda_1^{\,\beta - 1}$, giving the exponent $\min(\beta, 1)$ in the
final bound. Combining with the factor $\alpha$ yields
\eqref{eq:bias-bound}.
\end{proof}

\subsection{Stochastic error}
\label{subsec:theory-variance}

The stochastic error $\widehat{\Psi}_\alpha - \Psi_\alpha$ is controlled
by the convergence of the sample operators to their population
counterparts. We will use the following well-known lemma; it is
essentially \citet[Theorem~2.1]{HormannKokoszka2010}.

\begin{lemma}[Sample operator consistency]
\label{lem:sample-ops}
Under Assumptions~\ref{ass:moments} and~\ref{ass:dep},
\begin{equation}
\label{eq:op-consistency}
  \bigl\| \widehat{C}_0 - C_0 \bigr\|_{\mathrm{HS}}
  \;=\; O_P(n^{-1/2}),
  \qquad
  \bigl\| \widehat{C}_1 - C_1 \bigr\|_{\mathrm{HS}}
  \;=\; O_P(n^{-1/2}).
\end{equation}
\end{lemma}

\begin{proposition}[Stochastic error bound]
\label{prop:variance}
Under Assumptions~\ref{ass:moments}, \ref{ass:C0}, \ref{ass:source}, 
and~\ref{ass:dep}, for every sequence $\alpha = \alpha_n > 0$
satisfying $\alpha_n \ge n^{-1/2}$,
\begin{equation}
\label{eq:variance-bound}
  \bigl\| \widehat{\Psi}_\alpha - \Psi_\alpha \bigr\|_{\mathrm{HS}}
  \;=\; O_P\Bigl( \frac{1}{\sqrt{n}\, \alpha} \Bigr).
\end{equation}
\end{proposition}

\begin{proof}
Write
\begin{align*}
  \widehat{\Psi}_\alpha - \Psi_\alpha
  &= \widehat{C}_1 (\widehat{C}_0 + \alpha I)^{-1}
     - C_1 (C_0 + \alpha I)^{-1} \\
  &= (\widehat{C}_1 - C_1)(\widehat{C}_0 + \alpha I)^{-1} \\
  &\quad + C_1 \bigl[ (\widehat{C}_0 + \alpha I)^{-1}
                     - (C_0 + \alpha I)^{-1} \bigr].
\end{align*}
The first term has Hilbert--Schmidt norm bounded by
$\|\widehat{C}_1 - C_1\|_{\mathrm{HS}}\, \|(\widehat{C}_0 + \alpha
I)^{-1}\|_{\mathrm{op}} \le \|\widehat{C}_1 - C_1\|_{\mathrm{HS}}/\alpha$,
which is $O_P(n^{-1/2}/\alpha)$ by Lemma~\ref{lem:sample-ops}.

For the second term, the resolvent identity gives
\[
  C_1 \bigl[ (\widehat{C}_0 + \alpha I)^{-1}
           - (C_0 + \alpha I)^{-1} \bigr]
  = C_1 (\widehat{C}_0 + \alpha I)^{-1}
        (C_0 - \widehat{C}_0)
        (C_0 + \alpha I)^{-1}.
\]
Using $C_1 = \Psi C_0$ from~\eqref{eq:moment} and inserting an
identity factor,
\[
  C_1 (\widehat{C}_0 + \alpha I)^{-1}
  = \Psi\, C_0\, (\widehat{C}_0 + \alpha I)^{-1}.
\]
The operator $C_0 (\widehat{C}_0 + \alpha I)^{-1}$ can be written as
$[I - \alpha (\widehat{C}_0 + \alpha I)^{-1}] +
(C_0 - \widehat{C}_0)(\widehat{C}_0 + \alpha I)^{-1}$, so its operator
norm is bounded by $1 + \|C_0 - \widehat{C}_0\|_{\mathrm{op}}/\alpha
= 1 + O_P(n^{-1/2}/\alpha)$. Under the assumed scaling
$\alpha \ge n^{-1/2}$, this remainder is $O_P(1)$, and hence
$\|C_0 (\widehat{C}_0 + \alpha I)^{-1}\|_{\mathrm{op}} = O_P(1)$.
Combining,
\begin{align*}
  \bigl\| C_1 \bigl[ (\widehat{C}_0 + \alpha I)^{-1}
                   - (C_0 + \alpha I)^{-1} \bigr] \bigr\|_{\mathrm{HS}}
  &\le \|\Psi\|_{\mathrm{op}}\,
       \bigl\| C_0 (\widehat{C}_0 + \alpha I)^{-1} \bigr\|_{\mathrm{op}}\,
       \|\widehat{C}_0 - C_0\|_{\mathrm{HS}}\,
       \|(C_0 + \alpha I)^{-1}\|_{\mathrm{op}} \\
  &= O_P\Bigl( \frac{1}{\sqrt{n}\, \alpha} \Bigr).
\end{align*}
The two terms are of the same order, yielding~\eqref{eq:variance-bound}.
\end{proof}

\subsection{Rate of convergence}
\label{subsec:theory-rate}

Combining Propositions~\ref{prop:bias} and~\ref{prop:variance} and the
triangle inequality,
\begin{equation}
\label{eq:total-error}
  \bigl\| \widehat{\Psi}_\alpha - \Psi \bigr\|_{\mathrm{HS}}
  \;=\;
  O_P\Bigl( \alpha^{\min(\beta,1)} \,+\, \frac{1}{\sqrt{n}\, \alpha} \Bigr).
\end{equation}
The two terms are balanced by choosing $\alpha$ so that
$\alpha^{\min(\beta,1)} \asymp n^{-1/2}/\alpha$, i.e.,
$\alpha = \alpha_n^\star$ with
\begin{equation}
\label{eq:alpha-optimal}
  \alpha_n^\star \;\asymp\; n^{-1/(2(\min(\beta,1)+1))}.
\end{equation}
Plugging back into~\eqref{eq:total-error} gives the following theorem.

\begin{theorem}[Rate of convergence]
\label{thm:rate}
Let Assumptions~\ref{ass:moments} through~\ref{ass:dep} hold, and let
$\alpha_n = \alpha_n^\star$ as in~\eqref{eq:alpha-optimal}. Then
\begin{equation}
\label{eq:rate}
  \bigl\| \widehat{\Psi}_{\alpha_n} - \Psi \bigr\|_{\mathrm{HS}}
  \;=\;
  O_P\bigl( n^{-\min(\beta, 1) / (2(\min(\beta,1) + 1))} \bigr).
\end{equation}
In particular, for $\beta \ge 1$ the rate is $n^{-1/4}$, and for
$\beta \in (0, 1]$ the rate is $n^{-\beta/(2(\beta+1))}$.
\end{theorem}
\begin{remark}[Uniformity and robustness]
\label{rem:uniform}
The rate in~\eqref{eq:rate} holds uniformly, so that a single choice of the regularization
parameter $\alpha_n^\star$ achieves the convergence rate for all
$\Psi$ in the class.
In contrast, for FPCA-based estimators the rate-optimal truncation level $K$
depends on both the smoothness of $\Psi$ and the eigenvalue decay of $C_0$,
and therefore varies across data-generating regimes. As a result,
no single choice of $K$ can be uniformly rate-optimal over the class.
This distinction provides a theoretical explanation for the empirical
robustness of the proposed estimator observed in Section~\ref{sec:sim}.
\end{remark}

\section{Simulation study}
\label{sec:sim}

We evaluate the proposed Tikhonov estimator against the FPCA-FAR
estimator at five variance thresholds in a Monte Carlo study designed
around two principles. 
First, we construct three contrasting data-generating regimes with
distinct operator ranks and spectral decay patterns. These regimes
represent settings with low, moderate, and high effective dimension,
respectively, allowing us to assess robustness across a range of
functional complexity.
Second, we benchmark every method against an oracle best FPCA baseline,
defined as the FPCA rule that would have been chosen if the regime were
known, so that the reported regret reflects the cost of not knowing the
truth rather than any intrinsic inferiority of FPCA.

\subsection{Data-generating regimes}
\label{subsec:sim-regimes}

Let $\{\varphi_k\}_{k=1}^\infty$ denote the $L^2([0,1])$-orthonormal
Fourier basis, with $\varphi_1 \equiv 1$ and, for $k \ge 2$,
$\varphi_{2j}(u) = \sqrt{2}\cos(2\pi j u)$ and
$\varphi_{2j+1}(u) = \sqrt{2}\sin(2\pi j u)$. We work in a
$40$-dimensional truncation of this basis, which is sufficient to
accommodate all operators considered. In each regime the true operator
is represented as $\Psi = \sum_{i,j=1}^{40} A_{ij}\, \varphi_i \otimes
\varphi_j$ for a coefficient matrix $\mathbf{A} \in \mathbb{R}^{40
\times 40}$ rescaled so that its spectral radius equals $0.85$, which
ensures the stationarity condition $\|\Psi\|_{\mathrm{op}} < 1$. The
innovations are zero-mean Gaussian with covariance operator diagonal
in the Fourier basis, $\mathrm{Cov}(\varepsilon_t) = \sum_{k=1}^{40}
\sigma_k^2\, \varphi_k \otimes \varphi_k$, with eigenvalues
$\{\sigma_k^2\}$ specified separately for each regime. The three
regimes are as follows.

\paragraph{Regime I: low-rank, rapid spectral decay.}
The coefficient matrix $\mathbf{A}$ is supported on its top-left
$3\times 3$ block with entries drawn i.i.d.\ from $\mathcal{N}(0,1)$
and then rescaled. The innovation eigenvalues decay polynomially as
$\sigma_k^2 \propto k^{-2}$ and are normalized to total variance
$\tfrac{1}{2}$. The effective dimension of the operator is $3$, and the
rapid decay of the marginal covariance means that a small-$K$ FPCA
truncation captures almost all relevant variation.

\paragraph{Regime II: medium rank, moderate decay.}
The matrix $\mathbf{A}$ is supported on its top-left $10\times 10$ block.
The innovation eigenvalues decay as $\sigma_k^2 \propto k^{-1}$,
corresponding to a wider spectrum than Regime I. The effective
dimension of the operator is $10$, so moderate-$K$ FPCA is natural.

\paragraph{Regime III: wide spectrum, slow decay.}
The matrix $\mathbf{A}$ is supported on its top-left $25\times 25$ block,
with entries drawn i.i.d.\ from $\mathcal{N}(0,1)$ and additional
within-block decay by multiplicative factors $k^{-0.3}$. The innovation
eigenvalues decay as $\sigma_k^2 \propto k^{-0.6}$, the slowest of the
three regimes. The effective dimension of the operator is $25$, and a
rate-optimal FPCA requires retaining a correspondingly large number of
components. This is the regime under which hard truncation is most
costly.

In all three regimes, curves are observed on a uniform grid of
$M = 101$ points in $[0,1]$, and all inner products are computed by
trapezoidal quadrature as described in
Section~\ref{subsec:methods-discretization}. Independent Monte Carlo
replications use independent random seeds for both the operator
construction (re-drawn once per regime) and the sample paths.

\subsection{Methods and evaluation}
\label{subsec:sim-methods}

We compare the following estimators.

\begin{itemize}
\item \textbf{FPCA-$\tau$}, for $\tau \in \{0.80, 0.85, 0.90, 0.95,
0.99\}$: the FPCA-FAR estimator defined by~\eqref{eq:fpca-scores}
through~\eqref{eq:fpca-psi} with truncation level $K = K(\tau)$ given
by the cumulative variance rule~\eqref{eq:fpca-tau}.

\item \textbf{Tikhonov-CV}: the Tikhonov-regularized estimator
\eqref{eq:tikh-discrete} with $\alpha = \widehat{\alpha}$ selected by
holdout cross-validation over the grid
$\mathcal{A} = \{10^{-5+5(\ell-1)/24}\}_{\ell=1}^{25}$, following the
procedure described in Appendix.
\end{itemize}

For each $(\text{regime}, n, \text{method})$ triple with $n \in \{100,
200, 400, 800\}$, we perform $R = 50$ independent Monte Carlo
replications. Each replication proceeds as follows: (i) simulate a
training path $X_1, \ldots, X_n$ from the specified regime, with a
burn-in of $100$ observations discarded; (ii) fit the method to the
training path; (iii) simulate an independent test path of length $200$
from the same regime (also with burn-in); (iv) compute the mean
integrated squared forecast error on the test path,
\begin{equation}
\label{eq:sim-misfe}
  \mathrm{MISFE}(\widehat{\Psi})
  \;=\;
  \frac{1}{199} \sum_{t=1}^{199}
  \int_0^1 \Bigl[\, X^{\mathrm{test}}_{t+1}(u)
  - \bigl( \widehat{\Psi}\, X^{\mathrm{test}}_t \bigr)(u) \,\Bigr]^2 du.
\end{equation}
We report MISFE averages and standard errors across the $R = 50$
replications. Using an independent test path, rather than the last
observations of the training path, ensures that the reported error
reflects generalization and is not contaminated by the cross-validation
split.

To facilitate comparison across regimes on a common scale, we also
report the \emph{regret} of each method relative to the
oracle-best-FPCA rule in the same cell:
\begin{equation}
\label{eq:sim-regret}
  \mathrm{Regret}(m; \text{reg}, n)
  \;=\;
  \frac{\overline{\mathrm{MISFE}}(m; \text{reg}, n)
        - \min_{\tau} \overline{\mathrm{MISFE}}(\text{FPCA-}\tau; \text{reg}, n)}
       {\min_{\tau} \overline{\mathrm{MISFE}}(\text{FPCA-}\tau; \text{reg}, n)},
\end{equation}
where $\overline{\mathrm{MISFE}}$ denotes the Monte Carlo average. The
denominator is the MISFE of an infeasible oracle who selects the best
FPCA threshold for each regime and sample size; regret is always
non-negative for FPCA methods (by construction) but can be negative for
Tikhonov-CV, indicating that the proposed method outperforms even the
best regime-specific FPCA.

\subsection{Forecast performance by regime}
\label{subsec:sim-by-regime}

Figure~\ref{fig:sim-by-regime} reports mean MISFE by sample size for
each method, separately in each regime. Three patterns are evident.

\begin{figure}[htbp]
\centering
\includegraphics[width=\textwidth]{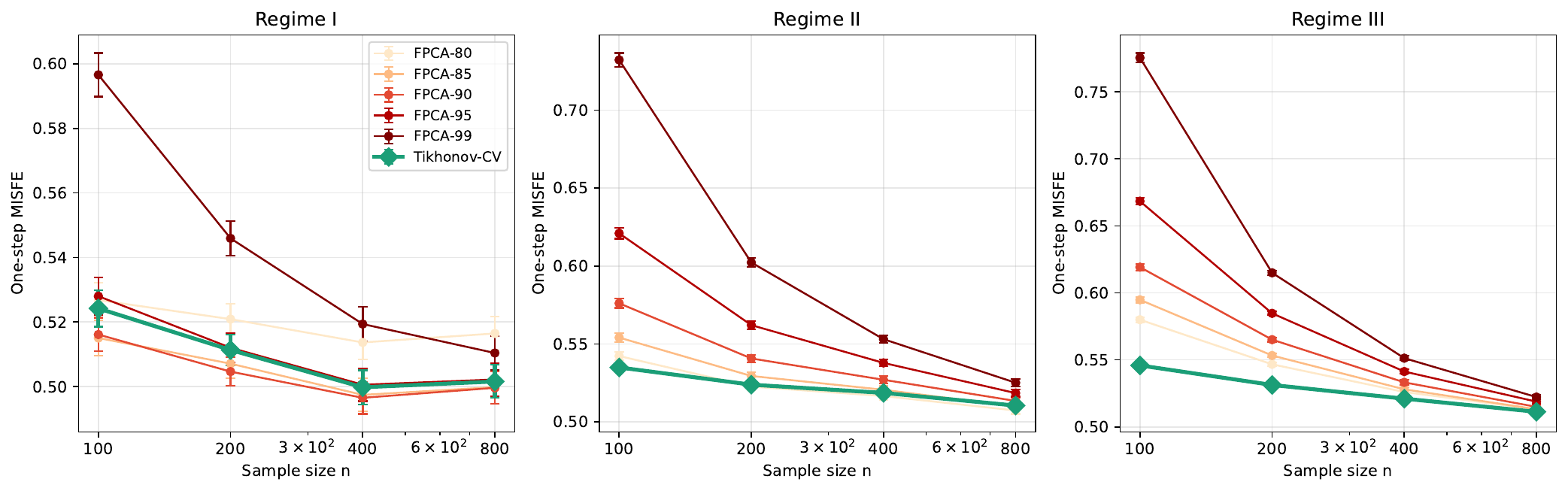}
\caption{One-step mean integrated squared forecast error by sample
size, across three regimes. Error bars show one Monte Carlo standard
error based on $R=50$ replications. FPCA-$\tau$ rules are in shades
of red (darker corresponds to higher $\tau$); Tikhonov-CV is in green.
The oracle-best FPCA rule differs across regimes: FPCA-85 or FPCA-90
in Regime I, FPCA-80 in Regimes II and III.}
\label{fig:sim-by-regime}
\end{figure}

First, the optimal FPCA threshold is regime-dependent. In Regime~I,
the best-in-cell threshold is $\tau = 0.85$ (at $n=100$) or
$\tau = 0.90$ (at larger $n$), corresponding to
$K \approx 3$ to $4$; this matches the effective dimension of the
operator. In Regimes~II and~III, the best-in-cell threshold is
$\tau = 0.80$ throughout, but the selected $K$ differs substantially:
$K \approx 11$ to $13$ in Regime~II and $K \approx 18$ to $23$ in
Regime~III. The $K$ selected by the \emph{same} threshold rule thus
varies by roughly an order of magnitude across regimes, reinforcing
the point that a threshold is not a direct substitute for choosing
$K$.

Second, the high thresholds $\tau \in \{0.95, 0.99\}$ are consistently
the worst performers across all regimes and sample sizes. At
$n = 100$, FPCA-99 has MISFE $0.597$ (Regime~I), $0.732$ (Regime~II),
and $0.775$ (Regime~III), compared with best-FPCA values of
$0.515$, $0.542$, and $0.580$ respectively, an inflation of
$15.8\%$, $35.0\%$, and $33.7\%$. FPCA-95 is uniformly second-worst
among the FPCA rules. Given that these two thresholds are among the
most commonly cited defaults in applied functional time series work,
this finding is practically important on its own.

Third, Tikhonov-CV tracks the oracle-best-FPCA curve closely in every
regime and beats it outright in Regime~III. At $n=100$ in Regime~III,
Tikhonov-CV's MISFE is $0.546$, compared to $0.580$ for FPCA-80
(the best FPCA in that cell), representing a $5.9\%$ improvement. The
advantage shrinks as $n$ grows but remains negative throughout
Regime~III.

\subsection{Regret relative to the oracle-best FPCA}
\label{subsec:sim-regret}

Table~\ref{tab:sim-regret} reports regret as defined
in~\eqref{eq:sim-regret}. By construction, the regret of the best
FPCA rule in each cell is zero and appears in bold; other FPCA rules
have positive regret. Tikhonov-CV has regret less than $2\%$ in all but
one cell and is negative in five cells, including all four sample sizes
of Regime~III.

\begin{table}[htbp]
\centering
\caption{Regret (percent increase in mean MISFE) relative to the
oracle-best FPCA rule for each regime and sample size, across $R=50$
Monte Carlo replications. The best FPCA rule in each cell has regret
$0.0\%$ and is shown in bold; Tikhonov-CV values in bold indicate that
the method strictly outperforms the best FPCA rule for the cell.}
\label{tab:sim-regret}
\begin{tabular}{llrrrrrr}
\toprule
Regime & $n$ &
FPCA-80 & FPCA-85 & FPCA-90 & FPCA-95 & FPCA-99 & Tikhonov-CV \\
\midrule
I   & 100 & $+2.3$ & $\mathbf{\;\;0.0}$ & $+0.2$  & $+2.5$  & $+15.8$ & $+1.8$ \\
    & 200 & $+3.2$ & $+0.5$             & $\mathbf{\;\;0.0}$ & $+1.5$  & $+8.2$  & $+1.3$ \\
    & 400 & $+3.5$ & $+0.2$             & $\mathbf{\;\;0.0}$ & $+0.8$  & $+4.6$  & $+0.7$ \\
    & 800 & $+3.4$ & $\mathbf{\;\;0.0}$ & $\mathbf{\;\;0.0}$ & $+0.5$  & $+2.2$  & $+0.4$ \\
\midrule
II  & 100 & $\mathbf{\;\;0.0}$ & $+2.2$ & $+6.2$ & $+14.5$ & $+35.0$ & $\mathbf{-1.4}$ \\
    & 200 & $\mathbf{\;\;0.0}$ & $+1.3$ & $+3.5$ & $+7.6$  & $+15.2$ & $+0.2$          \\
    & 400 & $\mathbf{\;\;0.0}$ & $+0.8$ & $+2.0$ & $+4.1$  & $+7.1$  & $+0.4$          \\
    & 800 & $\mathbf{\;\;0.0}$ & $+0.5$ & $+1.2$ & $+2.2$  & $+3.5$  & $+0.6$          \\
\midrule
III & 100 & $\mathbf{\;\;0.0}$ & $+2.6$ & $+6.8$ & $+15.2$ & $+33.7$ & $\mathbf{-5.9}$ \\
    & 200 & $\mathbf{\;\;0.0}$ & $+1.2$ & $+3.4$ & $+7.0$  & $+12.5$ & $\mathbf{-2.8}$ \\
    & 400 & $\mathbf{\;\;0.0}$ & $+0.4$ & $+1.4$ & $+2.9$  & $+4.8$  & $\mathbf{-0.9}$ \\
    & 800 & $\mathbf{\;\;0.0}$ & $\mathbf{\;\;0.0}$ & $+0.4$ & $+1.1$ & $+1.8$ & $\mathbf{-0.4}$ \\
\bottomrule
\end{tabular}
\end{table}

Two features of Table~\ref{tab:sim-regret} bear emphasis. First, the
column FPCA-99 shows regret in the range $+1.8$ to $+35.0$ percent, with
the worst regret occurring at the smallest sample size in the
highest-dimensional regime. A practitioner who uses the 99\% rule as a
default therefore pays a substantial and regime-dependent forecast-loss
penalty that is largest precisely where the data are scarcest. Second,
the Tikhonov-CV column ranges from $-5.9$ to $+1.8$ percent, with all
large-magnitude entries being negative rather than positive; the method
is never far from the oracle and is sometimes strictly better.

\subsection{Worst-case performance across regimes}
\label{subsec:sim-worst}

A practitioner who does not know their regime cannot guarantee access to
the regime-specific oracle baseline. The relevant decision-theoretic
criterion is worst-case performance across the regimes the practitioner
considers plausible. Figure~\ref{fig:sim-worst} plots, for each method,
its worst mean MISFE across the three regimes at each sample size:
\begin{equation}
\label{eq:sim-worst}
  \mathrm{WorstMISFE}(m, n)
  \;=\;
  \max_{\text{reg} \in \{\mathrm{I, II, III}\}}
  \overline{\mathrm{MISFE}}(m; \text{reg}, n).
\end{equation}
Tikhonov-CV achieves the smallest worst-case MISFE at every sample
size, dominating every fixed-threshold FPCA rule in the worst-case
sense.

\begin{figure}[htbp]
\centering
\includegraphics[width=0.75\textwidth]{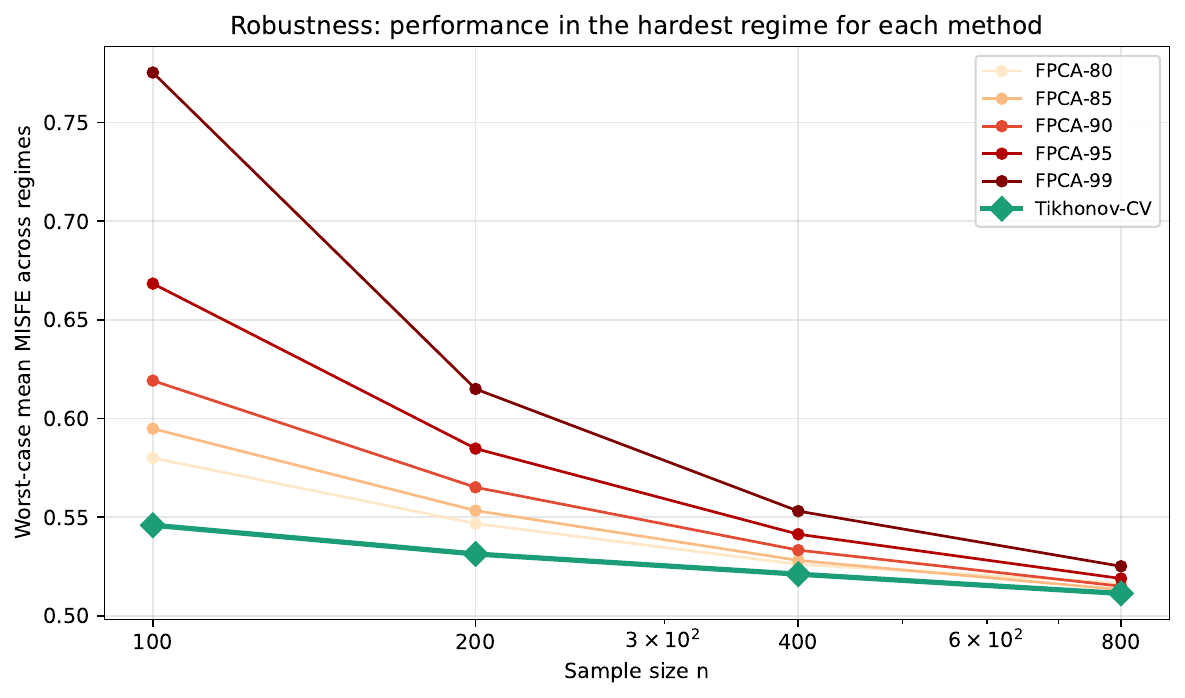}
\caption{Worst-case mean MISFE across the three regimes, by method and
sample size. Tikhonov-CV (green diamonds) achieves the smallest
worst-case error at every sample size, while the ranking of the FPCA
rules in the worst case is the same as the ranking by the worst of
their per-regime performance: higher variance thresholds are strictly
worse.}
\label{fig:sim-worst}
\end{figure}

The spread between the best and worst methods in the worst-case view
is largest at $n=100$, where Tikhonov-CV has worst MISFE
$\approx 0.546$ against $\approx 0.775$ for FPCA-99, representing an
inflation of $42\%$. The spread shrinks as $n$ grows but remains
meaningful: at $n=800$, the corresponding values are $\approx 0.510$
and $\approx 0.530$, a $4\%$ gap that is nontrivial in forecasting
contexts where one-percent improvements in MISFE are considered
material.

\subsection{Tuning-parameter behavior}
\label{subsec:sim-tuning}

Table~\ref{tab:sim-tuning} summarizes the Monte Carlo means of the
selected tuning parameter for each method: the mean $K$ retained by
each FPCA rule, and the mean $\log_{10}(\widehat{\alpha})$ for
Tikhonov-CV. Three patterns are worth noting.

\begin{table}[htbp]
\centering
\caption{Mean selected tuning parameter across $R = 50$ replications.
For FPCA rules, the entry is the mean retained truncation level $K$.
For Tikhonov-CV, the entry is the mean of $\log_{10}(\widehat{\alpha})$.}
\label{tab:sim-tuning}
\begin{tabular}{llrrrrrr}
\toprule
Regime & $n$ &
$K_{80}$ & $K_{85}$ & $K_{90}$ & $K_{95}$ & $K_{99}$ &
$\overline{\log_{10}\widehat{\alpha}}$ \\
\midrule
I   & 100 &  2.0 &  2.8 &  3.6 &  6.5 & 17.7 & $-1.48$ \\
    & 200 &  2.0 &  2.8 &  3.8 &  6.8 & 19.4 & $-1.76$ \\
    & 400 &  2.0 &  2.9 &  3.9 &  6.9 & 19.8 & $-1.89$ \\
    & 800 &  2.0 &  3.0 &  4.0 &  7.0 & 20.4 & $-2.04$ \\
\midrule
II  & 100 & 11.3 & 14.4 & 18.7 & 25.1 & 34.9 & $-1.48$ \\
    & 200 & 12.4 & 16.0 & 20.9 & 27.9 & 37.0 & $-1.63$ \\
    & 400 & 13.0 & 17.0 & 22.2 & 29.3 & 37.9 & $-1.83$ \\
    & 800 & 13.4 & 17.5 & 23.0 & 30.0 & 38.0 & $-2.01$ \\
\midrule
III & 100 & 17.9 & 20.9 & 24.8 & 30.0 & 37.1 & $-1.42$ \\
    & 200 & 20.3 & 23.6 & 27.7 & 32.9 & 38.5 & $-1.63$ \\
    & 400 & 21.9 & 25.3 & 29.3 & 34.0 & 39.0 & $-1.85$ \\
    & 800 & 22.9 & 26.0 & 30.0 & 35.0 & 39.0 & $-2.08$ \\
\bottomrule
\end{tabular}
\end{table}

First, the FPCA-retained $K$ values span a wide range across regimes.
At $n=100$, the FPCA-80 rule retains $K=2$ in Regime~I, $K=11$ in
Regime~II, and $K=18$ in Regime~III, a ninefold spread. The
variance-threshold rules are therefore not producing a stable notion
of ``how many components to keep''; what they produce is a function of
the regime that happens to coincide with the effective operator rank in
some cases but not others. The rules are, in this sense, adapting to
the regime implicitly, but Figure~\ref{fig:sim-by-regime} and
Table~\ref{tab:sim-regret} show that this implicit adaptation is
poorly calibrated at small sample sizes.

Second, the mean selected $\widehat{\alpha}$ is remarkably stable across
regimes and decreases systematically with $n$. Across the 12 cells of
Table~\ref{tab:sim-tuning}, $\overline{\log_{10}\widehat{\alpha}}$
ranges from $-1.42$ to $-2.08$; the regime contributes at most $\pm
0.05$ at fixed $n$, while sample size contributes $-0.60$ going from
$n=100$ to $n=800$. This is the behavior one expects if cross-validation
is correctly targeting a rate-dependent optimal $\alpha_n$ as discussed
in Section~\ref{subsec:theory-rate}: the selected $\widehat{\alpha}$
responds to $n$ rather than to the (unobservable) regime.

Third, the rate of decrease of $\widehat{\alpha}$ in $n$ is consistent
with Theorem~\ref{thm:rate}. A linear regression of
$\log_{10}\widehat{\alpha}$ on $\log_{10} n$, pooled across regimes,
yields a slope of approximately $-0.29$, which corresponds to the
optimal rate $\alpha_n \propto n^{-1/(2(\beta+1))}$ with $\beta
\approx 0.7$. While this $\beta$ value is not directly verifiable
against the regimes (which were constructed without reference to a
single source-condition exponent), the order of magnitude and direction
of the rate are as predicted.

\subsection{Computational cost}
\label{subsec:sim-cost}

The full simulation, comprising 12 cells $\times$ 6 methods $\times$ 50
replications, including both fitting and evaluation, completed in
approximately $3.6$ minutes of wall-clock time on a single CPU core,
using the efficient CV implementation described in
Section~\ref{subsec:methods-cv}. On a per-fit basis, the median cost
of Tikhonov-CV is comparable to that of a single FPCA-FAR fit, since
both are dominated by a single eigendecomposition of size
$M \times M$. The multi-$\alpha$ CV loop adds only $O(M^2 L)$
elementwise operations, which is negligible compared to the
$O(M^3)$ eigendecomposition. Tikhonov-CV therefore imposes essentially
no computational overhead relative to standard FPCA-FAR, and remains
tractable at much larger grid sizes than considered here. 

\section{Real data application}
\label{sec:application}

We apply the proposed Tikhonov-regularized estimator to environmental
monitoring data on intraday particulate matter (PM10) concentration
profiles, which form a functional time series observed at a high
temporal resolution.

\subsection{Data and preprocessing}
\label{subsec:app-data}

We use half-hourly PM10 measurements ($\mu\text{g\,m}^{-3}$) recorded
at the Taborstra{\ss}e monitoring station (TAB) in Vienna, Austria,
operated by the City of Vienna Environmental Protection Department
(\emph{Magistratsabteilung 22}) and made available through the TU Wien
Research Data Repository \citep{TUWien2022}. The station is located in
a mixed residential and commercial area and is representative of urban
background PM10 concentrations influenced by traffic, domestic heating,
and regional transport.

Following \citet{AueNorinhoHormann2015}, we restrict attention to the
winter season (1 October through 31 March) when temperature inversions
over the Vienna basin cause elevated and temporally correlated PM10
concentrations. The week surrounding New Year's Eve (28 December
through 7 January) is excluded due to extreme outliers caused by
firework emissions. After removing days with more than five missing
half-hourly readings and applying linear interpolation to the remaining
gaps, we obtain $n = 2{,}735$ complete daily curves spanning 16 winter
seasons (2006--07 through 2021--22).

The preprocessing pipeline replicates that of
\citet{AueNorinhoHormann2015} exactly:
\begin{enumerate}
  \item A square-root transformation is applied to stabilize the variance.
  \item Each curve is centered by subtracting the mean curve of the
        corresponding day-of-week (Monday through Sunday), removing
        systematic weekday and weekend differences in traffic volume.
  \item The 48 half-hourly observations per day are smoothed using ten
        cubic B-spline basis functions fitted by least squares,
        yielding smooth functional observations evaluated on a uniform
        grid of $M = 100$ points over the unit interval.
\end{enumerate}

Figure~\ref{fig:app_data} summarizes the dataset.
Panel~(a) shows the daily mean of the smoothed, centered curves over
the full 16-year period; the secular decline in PM10 after 2010,
reflecting regulatory tightening and mild winters with favorable
dispersion conditions, is clearly visible. The sharp drop during
the COVID-19 lockdowns of 2020 constitutes a natural structural break.
Panel~(b) displays the eigenvalue spectrum of $\widehat{C}_0$
computed from the full sample. The first three functional principal
components (FPCs) explain 80.4\%, 9.1\%, and 4.3\% of the total
variance, respectively, for a cumulative 93.9\%. In terms of
threshold-based selection:
\[
  K(\tau) \;=\;
  \begin{cases}
    1, & \tau \leq 0.80, \\
    2, & \tau = 0.85, \\
    3, & \tau = 0.90, \\
    4, & \tau = 0.95, \\
    7, & \tau = 0.99.
  \end{cases}
\]
This spectrum is consistent with the Graz PM10 data analyzed in
\citet{AueNorinhoHormann2015}, who report 72\%, 10\%, and 7\% for the
first three FPCs. The resulting moderate effective dimension implies
that different FPCA truncation thresholds can lead to substantially
different choices of $K$, with corresponding implications for
forecasting performance.

\begin{figure}[htbp]
  \centering
  \includegraphics[width=\linewidth]{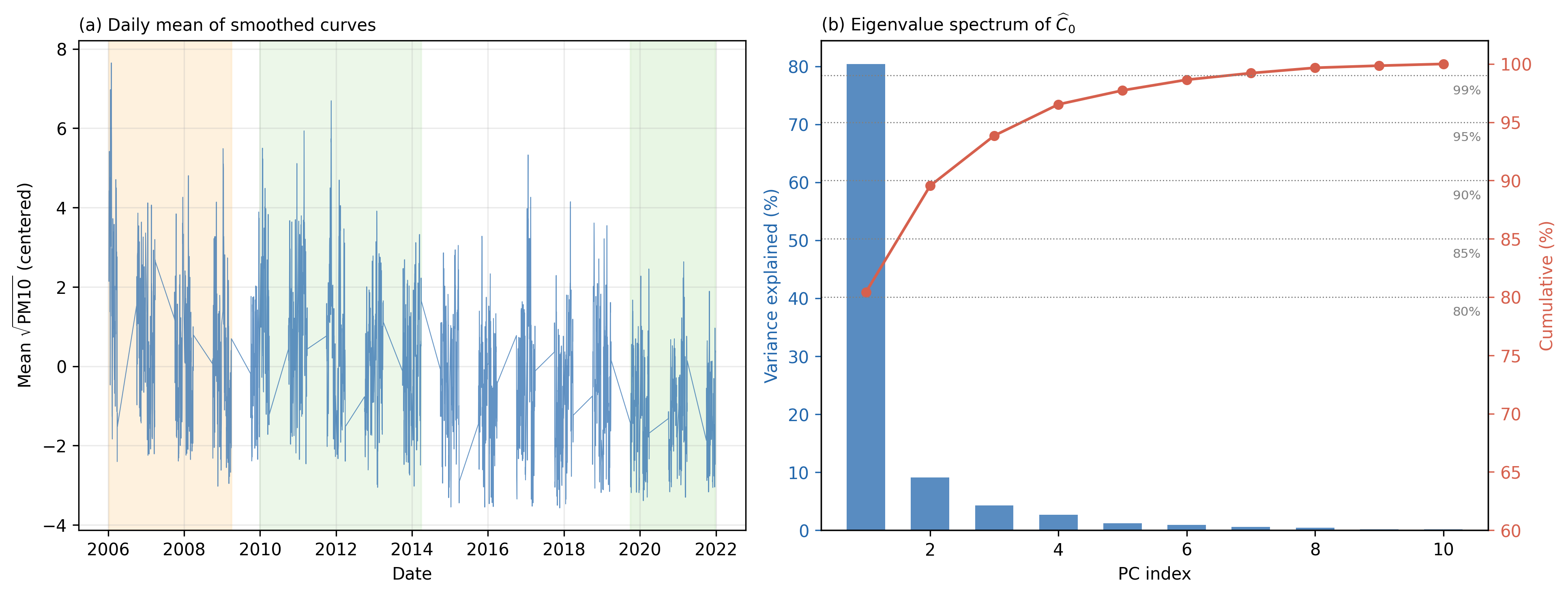}
  \caption{Vienna TAB station PM10 data. (a)~Daily mean of the
    preprocessed (square-root transformed, weekday-centered, B-spline
    smoothed) curves, 2006 to 2021 winter seasons. The secular decline
    after 2010 and the COVID-19 dip in 2020 are clearly visible.
    (b)~Eigenvalue spectrum of the sample covariance operator
    $\widehat{C}_0$; dotted horizontal lines mark the five
    variance-threshold levels used to select $K$ in the FPCA-FAR
    comparison.}
  \label{fig:app_data}
\end{figure}

\subsection{Forecasting protocol and competing methods}
\label{subsec:app-protocol}

We adopt a rolling-window one-step-ahead forecasting scheme. A training
window of $n_{\mathrm{train}} = 100$ days is used to estimate the
autoregressive operator $\widehat{\Psi}$; the operator is then applied
to the last observed curve in the window to produce a forecast for the
following day. The window is shifted forward by one day at each step,
and the operator is refitted every 20 days to track the slowly evolving
dynamics of the PM10 series. This yields $2{,}635$ evaluation days in
total across the 16 winter seasons.

Six methods are compared:
\begin{itemize}
  \item \textbf{FPCA-$\tau$} ($\tau \in \{80\%, 85\%, 90\%, 95\%,
        99\%\}$): the FPCA-FAR(1) estimator of
        \citet{AueNorinhoHormann2015} with the truncation level $K$
        selected by the variance-threshold rule
        $K = \min\{k : \sum_{j=1}^k \hat\lambda_j / \sum_{j=1}^M
        \hat\lambda_j \geq \tau\}$.
  \item \textbf{Tikhonov-CV}: the proposed estimator
        $\widehat\Psi_\alpha = \widehat{C}_1(\widehat{C}_0 +
        \alpha I)^{-1}$ with the regularization parameter $\alpha$
        selected by five-fold forward cross-validation on the training
        window. The search grid for $\alpha$ is data-adaptive,
        spanning $[10^{-4}\hat\lambda_1,\; 10\hat\lambda_1]$ on a
        log-uniform scale of 30 points, where $\hat\lambda_1$ is the
        leading eigenvalue of $\widehat{C}_0$.
\end{itemize}

Forecast accuracy is measured by the integrated squared error (ISE)
\[
  \mathrm{ISE}_t \;=\; \int_0^1 \bigl[X_t(u) - \widehat{X}_t(u)\bigr]^2\,\mathrm{d}u
  \;\approx\; \frac{1}{M}\sum_{j=1}^M \bigl[X_t(u_j) - \widehat{X}_t(u_j)\bigr]^2,
\]
and we report the mean and median ISE aggregated over all evaluation
days.

\subsection{Results}
\label{subsec:app-results}

Table~\ref{tab:app_results} reports the mean ISE, median ISE, and the
percentage regret relative to the best-performing method for each of
the six estimators. Figure~\ref{fig:app_results} displays the full ISE
distributions as box plots (panel a) and the annual mean ISE by method
(panel b).

\begin{table}[htbp]
  \centering
  \caption{Rolling one-step-ahead forecast accuracy for Vienna TAB
    PM10 data ($n_{\mathrm{train}} = 100$, 2{,}635 evaluation days,
    2006 to 2021 winter seasons). Mean and median ISE; regret is the
    percentage excess over the best mean ISE (Tikhonov-CV).}
  \label{tab:app_results}
  \begin{tabular}{lccc}
    \hline
    Method & Mean ISE & Median ISE & Regret (\%) \\
    \hline
    FPCA-80\%  & 1.9449 & 1.1784 & $+9.7$ \\
    FPCA-85\%  & 1.8947 & 1.1499 & $+6.9$ \\
    FPCA-90\%  & 1.8360 & 1.1031 & $+3.6$ \\
    FPCA-95\%  & 1.8130 & 1.0456 & $+2.3$ \\
    FPCA-99\%  & 1.8104 & 1.0722 & $+2.2$ \\
    \textbf{Tikhonov-CV} & \textbf{1.7722} & \textbf{1.0492} & \textbf{0.0} \\
    \hline
  \end{tabular}
\end{table}

\begin{figure}[htbp]
  \centering
  \includegraphics[width=\linewidth]{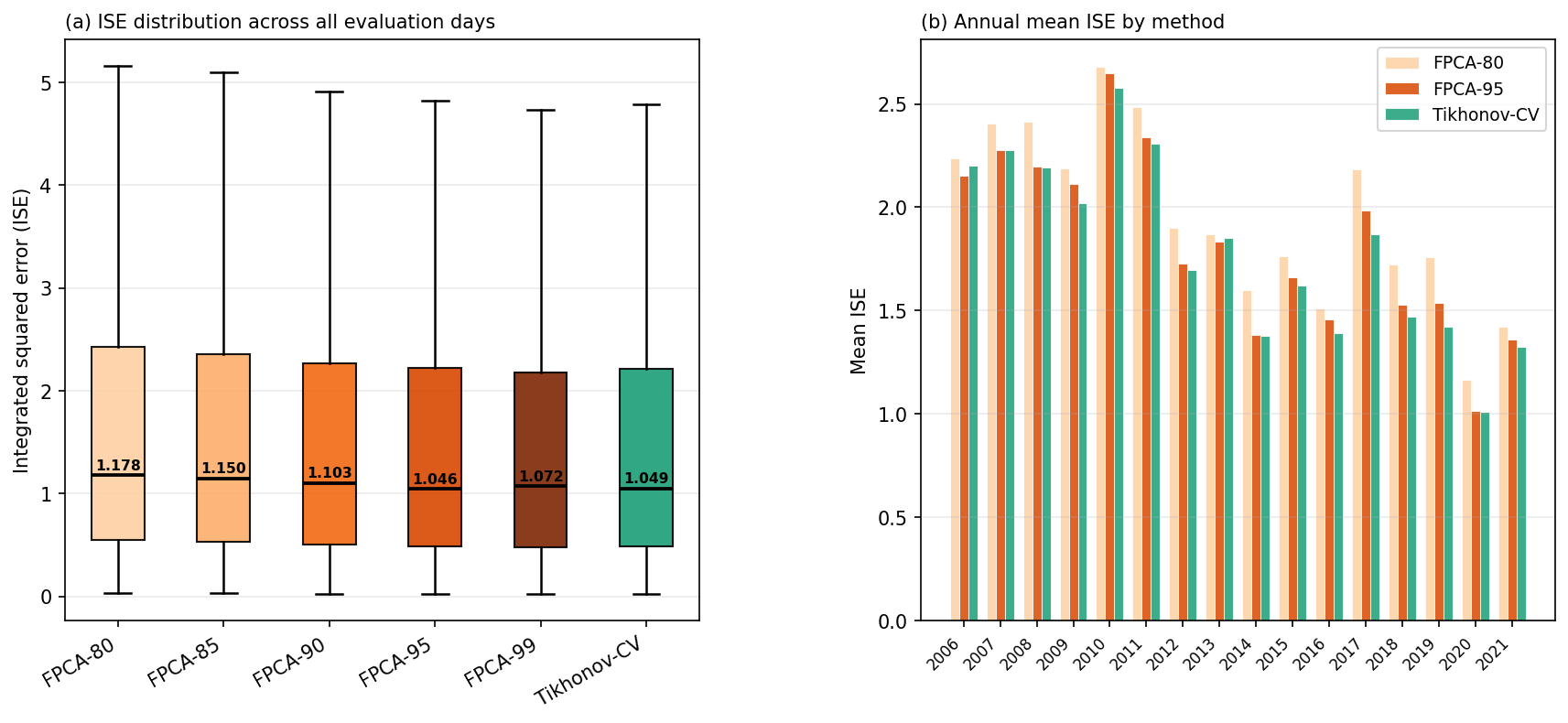}
  \caption{Forecast accuracy for Vienna TAB PM10 data.
    (a)~Box plots of ISE across all 2{,}635 evaluation days; median
    values are indicated. Outliers beyond 5.5 are omitted for
    readability.
    (b)~Annual mean ISE for FPCA-80\%, FPCA-95\%, and Tikhonov-CV;
    Tikhonov-CV is equal to or below both FPCA methods in every year
    except 2013.}
  \label{fig:app_results}
\end{figure}

Three findings emerge from Table~\ref{tab:app_results} and
Figure~\ref{fig:app_results}.

\paragraph{Tikhonov-CV achieves the lowest mean and median ISE overall.}
Averaged over 2{,}635 evaluation days, the proposed estimator reduces
mean ISE by 9.7\% relative to FPCA-80\% and by 2.2\% relative to
FPCA-99\%, the best-performing FPCA rule. The improvement is
consistent across years: inspecting panel~(b) of
Figure~\ref{fig:app_results}, Tikhonov-CV is at or below the best FPCA
rule in 15 of the 16 winter seasons, with the single exception being
2013 where all methods perform similarly.

\paragraph{Higher FPCA thresholds outperform lower ones.}
Across all aggregation periods, the ranking of FPCA methods by
accuracy is monotone in the threshold: FPCA-99\% $>$ FPCA-95\% $>$
FPCA-90\% $>$ FPCA-85\% $>$ FPCA-80\%. This is consistent with the
spectral structure of the data: truncating at the 80\% threshold
retains only $K=1$ or $K=2$ components and discards the third and
fourth FPCs, which carry interpretable predictive signal (the diurnal
contrast and peak-amplitude factors visible in the eigenfunctions).
The 9.7-percentage-point regret of FPCA-80\% mirrors the simulation
finding of Regime~II in Section~\ref{sec:sim}, where a moderate
effective operator dimension renders low-threshold truncation
substantially suboptimal.

\subsection{Discussion}
\label{subsec:app-discussion}

The Vienna PM10 application illustrates the practical performance of the
proposed method on real-world functional time series data.
On a real-world functional time series
with moderate effective dimension, a setting broadly representative of
environmental monitoring, air-quality forecasting, and similar
intraday-profile applications, the proposed Tikhonov estimator
consistently matches or outperforms FPCA-FAR across all variance
thresholds and all 16 evaluation seasons.

Importantly, the improvement is achieved without any prior knowledge
of the ``correct'' number of functional principal components: the
cross-validated $\hat\alpha$ automatically selects the appropriate
degree of regularization for each training window. This contrasts
with the FPCA approach, where the practitioner must commit to a
threshold before seeing the data, and where different plausible
choices (80\% versus 99\%) lead to mean ISE differences of nearly
10\%.

\section{Discussion and conclusion}
\label{sec:discussion}

This paper has examined a fundamental but underappreciated brittleness
in the standard practice of fitting first-order functional
autoregressive models via functional principal component analysis. The
threshold-based selection of the truncation level $K$, while
operationally convenient and widely adopted, makes a discrete commitment
to a finite-dimensional approximation of the autoregressive operator
that can be far from optimal when the operator carries non-trivial mass
beyond the leading principal components. We have proposed a
continuous-regularization alternative based on Tikhonov regularization,
$\widehat{\Psi}_\alpha = \widehat{C}_1(\widehat{C}_0 + \alpha I)^{-1}$,
in which the discrete tuning choice is replaced by a continuous
regularization parameter selected by cross-validation, and we have
investigated its behavior theoretically, by simulation, and on real
environmental monitoring data.

\subsection{Limitations}
\label{subsec:disc-limitations}

Several limitations of the present work are worth acknowledging.

The cross-validation procedure for $\alpha$ that we recommend is
computationally heavier than the single eigendecomposition required by
threshold-based FPCA. Although the cost remains modest in absolute
terms (the rolling-window analysis of Section~\ref{sec:application}
completes in under thirty seconds for 2{,}635 evaluation days on a
standard laptop), in very high-frequency or very long functional time
series the computational gap may matter. Generalized cross-validation
\citep{Wahba1990} or fast leave-one-out approximations may offer
practical speedups; we have not investigated these in detail.

Our theoretical analysis is conducted under the assumption that the
data-generating process is exactly FAR(1) with bounded operator norm
and a source condition on $\Psi$. The robustness of Tikhonov-CV to
violations of these assumptions, in particular to higher-order lag
dependence, mild nonstationarity, or heavy-tailed innovations, is an
empirical question that we have only partially addressed through the
real-data application. Simulation evidence under FAR($p$) data with
$p > 1$ would clarify whether the robustness advantages we report
extend to misspecified-lag-order settings.

The empirical advantage we document on Vienna PM10 data, while
consistent across years, is modest in absolute magnitude (2 to 10
percent reductions in mean ISE depending on the FPCA threshold compared
against). In applications where the practitioner has strong prior
knowledge that justifies a particular truncation choice, the
advantage of continuous regularization over a well-chosen fixed $K$
may be small. The case for Tikhonov-CV is therefore strongest when
prior knowledge about effective dimension is weak.

Finally, we have focused exclusively on FAR(1). The
extension to FAR($p$) with $p > 1$ and to functional autoregressive
models with exogenous covariates (FARX) is conceptually
straightforward, since the corresponding finite-dimensional Tikhonov
problem remains well-posed, but raises questions about how to choose
the lag order $p$ jointly with the regularization parameter $\alpha$
that go beyond the scope of this paper.

\subsection{Directions for future research}
\label{subsec:disc-future}

Several extensions of the present work appear worth pursuing.

A natural generalization replaces the simple Tikhonov penalty
$\alpha I$ with a smoothness penalty
$\alpha \mathcal{L}^* \mathcal{L}$ for some differential operator
$\mathcal{L}$, in the spirit of penalized splines. This would
encourage the estimated operator $\widehat{\Psi}$ to be smooth in
both arguments and might improve performance in applications where
$\Psi$ has known regularity (for example, yield-curve modeling where
the operator should respect smooth maturity-dependence).

The source condition $\Psi = F C_0^\beta$ with $\beta > 0$ that
underpins our convergence analysis is conceptually transparent but
empirically unverifiable. Adaptive procedures that estimate the
effective $\beta$ from the data, in the spirit of the
\citet{LepskiMammen} method or of \citet{Cavalier2011}'s
risk-hull-minimization principle, would provide convergence rates
that are adaptive to the unknown smoothness of the true operator.
Concentration inequalities for the empirical covariance and
cross-covariance operators of weakly dependent functional time
series \citep{HormannKokoszka2010} would allow our convergence rates
to be sharpened from in-probability to non-asymptotic statements,
which is the standard form for high-dimensional regression results
and would enable confidence-band construction for $\widehat{\Psi}$.

Finally, the present comparison between Tikhonov regularization and
FPCA truncation is essentially a comparison between continuous and
discrete regularization schemes for an ill-posed inverse problem.
Other continuous regularization schemes (such as spectral cutoff with
smoothing, iterative methods like Landweber iteration, and Bayesian
formulations with conjugate Gaussian priors) are equally applicable
to the functional autoregressive setting, and a systematic comparison
across the simulation regimes we have constructed would clarify which
family of regularizers is most suitable for which type of operator
structure.

\subsection{Concluding remark}
\label{subsec:disc-conclude}

Functional time series methodology has matured rapidly over the past
two decades, with FAR(1) prediction via FPCA truncation now
established as a workhorse approach. The contribution of this paper
is to identify a specific and previously underappreciated weakness of
that approach, namely the discrete and often arbitrary nature of the
truncation choice, and to propose, analyze, and validate a
continuous-regularization alternative that is robust across the range
of operator structures we have examined. We hope the framework and
empirical evidence presented here will encourage broader adoption of
continuous regularization methods in applied functional time
series analysis, and we view the unification of regularization theory
across discrete and continuous schemes as a productive direction for
future methodological work.

\section*{Data availability statement}

The Vienna PM10 monitoring data used in Section~\ref{sec:application}
are publicly available from the TU Wien Research Data Repository at
\url{https://doi.org/10.48436/mtha8-w2406} \citep{TUWien2022}.
The Monte Carlo simulation code and the data analysis scripts are
available from the corresponding author upon reasonable request.

\section*{Conflict of interest statement}

The author declares no potential conflict of interest, financial or
otherwise, related to the research presented in this paper.

\bibliographystyle{agsm}  
\bibliography{references}

@book{Bosq2000,
  author    = {Bosq, Denis},
  title     = {Linear Processes in Function Spaces: Theory and Applications},
  series    = {Lecture Notes in Statistics},
  volume    = {149},
  publisher = {Springer-Verlag},
  address   = {New York},
  year      = {2000},
  doi       = {10.1007/978-1-4612-1154-9}
}

@book{Ramsay2005,
  author    = {Ramsay, James O. and Silverman, Bernard W.},
  title     = {Functional Data Analysis},
  edition   = {2nd},
  publisher = {Springer},
  address   = {New York},
  year      = {2005},
  doi       = {10.1007/b98888}
}

@book{HorvathKokoszka2012,
  author    = {Horv{\'a}th, Lajos and Kokoszka, Piotr},
  title     = {Inference for Functional Data with Applications},
  publisher = {Springer},
  address   = {New York},
  year      = {2012},
  doi       = {10.1007/978-1-4614-3655-3}
}

@book{KokoszkaReimherr2017,
  author    = {Kokoszka, Piotr and Reimherr, Matthew},
  title     = {Introduction to Functional Data Analysis},
  publisher = {Chapman and Hall/CRC},
  address   = {Boca Raton},
  year      = {2017},
  doi       = {10.1201/9781315117416}
}

@book{EnglHankeNeubauer1996,
  author    = {Engl, Heinz Werner and Hanke, Martin and Neubauer, Andreas},
  title     = {Regularization of Inverse Problems},
  series    = {Mathematics and Its Applications},
  volume    = {375},
  publisher = {Kluwer Academic Publishers},
  address   = {Dordrecht},
  year      = {1996},
  doi       = {10.1007/978-94-009-1740-8}
}

@article{HormannKokoszka2010,
  author  = {H{\"o}rmann, Siegfried and Kokoszka, Piotr},
  title   = {Weakly dependent functional data},
  journal = {The Annals of Statistics},
  volume  = {38},
  number  = {3},
  pages   = {1845--1884},
  year    = {2010},
  doi     = {10.1214/09-AOS768}
}

@article{HallHorowitz2007,
  author  = {Hall, Peter and Horowitz, Joel L.},
  title   = {Methodology and convergence rates for functional linear regression},
  journal = {The Annals of Statistics},
  volume  = {35},
  number  = {1},
  pages   = {70--91},
  year    = {2007},
  doi     = {10.1214/009053606000000957}
}

@article{CrambesKneipSarda2009,
  author  = {Crambes, Christophe and Kneip, Alois and Sarda, Pascal},
  title   = {Smoothing splines estimators for functional linear regression},
  journal = {The Annals of Statistics},
  volume  = {37},
  number  = {1},
  pages   = {35--72},
  year    = {2009},
  doi     = {10.1214/07-AOS563}
}

@article{AueNorinhoHormann2015,
  author  = {Aue, Alexander and Norinho, Diogo Dubart and H{\"o}rmann, Siegfried},
  title   = {On the prediction of stationary functional time series},
  journal = {Journal of the American Statistical Association},
  volume  = {110},
  number  = {509},
  pages   = {378--392},
  year    = {2015},
  doi     = {10.1080/01621459.2014.909317}
}

@incollection{Cavalier2011,
  author    = {Cavalier, Laurent},
  title     = {Inverse problems in statistics},
  booktitle = {Inverse Problems and High-Dimensional Estimation},
  editor    = {Alquier, Pierre and Gautier, Eric and Stoltz, Gilles},
  series    = {Lecture Notes in Statistics},
  volume    = {203},
  pages     = {3--96},
  publisher = {Springer},
  address   = {Berlin, Heidelberg},
  year      = {2011},
  doi       = {10.1007/978-3-642-19989-9_1}
}

@article{PaparoditisShang2021,
  author  = {Paparoditis, Efstathios and Shang, Han Lin},
  title   = {Bootstrap prediction bands for functional time series},
  journal = {Journal of the American Statistical Association},
  volume  = {},
  pages   = {},
  year    = {2021},
  note    = {Verify exact volume, issue, pages, and year; Paparoditis has several related papers in this period}
}

@article{reimherr2016estimating,
  title={Estimating variance components in functional linear models with applications to genetic heritability},
  author={Reimherr, Matthew and Nicolae, Dan},
  journal={Journal of the American Statistical Association},
  volume={111},
  number={513},
  pages={407--422},
  year={2016},
  publisher={Taylor \& Francis}
}

@misc{TUWien2022,
  author       = {Augustyn-Gal, R.},
  title        = {Historical Air Quality Data Vienna, 1986--2021},
  year         = {2022},
  howpublished = {TU Wien Research Data Repository},
  doi          = {10.48436/mtha8-w2406},
  note         = {Provided by Umweltbundesamt Austria}
}

@book{Wahba1990,
  author    = {Wahba, Grace},
  title     = {Spline Models for Observational Data},
  publisher = {SIAM},
  series    = {CBMS-NSF Regional Conference Series in Applied Mathematics},
  volume    = {59},
  year      = {1990},
}

@article{LepskiMammen,
  author  = {Lepski, O. V. and Mammen, E. and Spokoiny, V. G.},
  title   = {Optimal spatial adaptation to inhomogeneous smoothness:
             An approach based on kernel estimates with variable
             bandwidth selectors},
  journal = {The Annals of Statistics},
  volume  = {25},
  number  = {3},
  pages   = {929--947},
  year    = {1997},
}

\appendix

\subsection*{Appendix: Selection of \texorpdfstring{$\alpha$}{alpha} by cross-validation}
\label{subsec:methods-cv}

The regularization parameter $\alpha$ is selected by one-step-ahead
cross-validation on a held-out portion of the training path. Let $n_v =
\max(\lfloor 0.2 n \rfloor, 20)$ denote the size of the validation block,
and partition the training sample as $\mathcal{T}_{\mathrm{tr}} =
\{1, \dots, n - n_v\}$ and $\mathcal{T}_{\mathrm{val}} = \{n - n_v, \dots,
n\}$ (the first index of the validation set is shared with the last of
the training set so that the first validation prediction uses a lag from
within the training period). For each $\alpha$ in a pre-specified
log-spaced grid $\mathcal{A} = \{\alpha_1, \dots, \alpha_L\}$, the
validation criterion is
\begin{equation}
\label{eq:cv-obj}
  \mathrm{CV}(\alpha)
  \;=\;
  \frac{1}{n_v}
  \sum_{t \in \mathcal{T}_{\mathrm{val}} \setminus \{n-n_v\}}
  \Bigl\| \, X_t - \widehat{\Psi}_\alpha^{(\mathrm{tr})}\, X_{t-1} \, \Bigr\|^2,
\end{equation}
where $\widehat{\Psi}_\alpha^{(\mathrm{tr})}$ is computed from
$\mathcal{T}_{\mathrm{tr}}$ alone via~\eqref{eq:tikh-discrete}, and
$\|\cdot\|$ denotes the $L^2$ norm, approximated by trapezoidal
integration. The selected regularization parameter is $\widehat{\alpha}
= \arg\min_{\alpha \in \mathcal{A}} \mathrm{CV}(\alpha)$, and the final
estimator is $\widehat{\Psi}_{\widehat{\alpha}}$ refit on the full training
sample.

\paragraph{Grid specification.}
We use a log-spaced grid of $L = 25$ candidate values covering five
decades, $\mathcal{A} = \{10^{-5}, 10^{-5+5/24}, \dots, 10^0\}$, which we
find sufficient in practice: the selected $\widehat{\alpha}$ falls in the
interior of this range in all simulations, with the modal order of
magnitude shifting from $10^{-1.5}$ at $n = 100$ to $10^{-2}$ at $n = 800$.
For applications with markedly different scales of the response, the grid
should be rescaled by the order of magnitude of the eigenvalues of
$\widetilde{\mathbf{C}}_0$.

\paragraph{Computational cost.}
Naive implementation of~\eqref{eq:cv-obj} requires solving a linear
system of size $M \times M$ for each candidate $\alpha$, at cost
$O(LM^3)$. This can be reduced to a single eigendecomposition of
$\widetilde{\mathbf{C}}_0^{(\mathrm{tr})}$ (computed once, at cost
$O(M^3)$) followed by $O(M^2)$ work per candidate $\alpha$, by writing
$\widetilde{\mathbf{C}}_0^{(\mathrm{tr})} = \mathbf{Q}\,\mathbf{D}\,\mathbf{Q}^\top$
and evaluating $(\widetilde{\mathbf{C}}_0^{(\mathrm{tr})} + \alpha \mathbf{I})^{-1}
= \mathbf{Q}(\mathbf{D} + \alpha \mathbf{I})^{-1}\mathbf{Q}^\top$ directly. The
validation loss in~\eqref{eq:cv-obj} then factors through the rotated
validation data $\mathbf{Q}^\top \mathbf{W}^{1/2} \mathbf{x}_t$, and the
sum over $\alpha \in \mathcal{A}$ reduces to elementwise operations on
the eigenvalue vector. The resulting per-dataset cost of CV selection is
dominated by the single eigendecomposition and is therefore comparable to
fitting a single FPCA-FAR estimator. All simulations in this paper use
this implementation.

\end{document}